\documentclass{llncs}

\usepackage{amsmath,amssymb,amsfonts}
\input{epsf}

\newcommand {\ignore} [1] {}

\newcommand{\ssn}  {\sf Star-SN}

\newcommand{\sna}  {\sf Star-NA}
\newcommand{\cna}  {\sf Centered-NA}

\newtheorem{observation}[corollary]{Observation}

\begin{document}

\title{Approximating Source Location and \\ Star Survivable Network Problems}

\author{Guy Kortsarz\inst{1} \and Zeev Nutov\inst{2}}
\institute{Rutgers University, Camden    \email{guyk@camden.rutgers.edu} \and 
           The Open University of Israel \email{nutov@openu.ac.il}}

\maketitle

% Key-words: Source location, Rooted connectivity, Connectivity augmentation, Submodular cover

\begin{abstract}
In {\sf Source Location} ({\sf SL}) problems 
the goal is to select a mini\-mum cost source set $S \subseteq V$
such that the connectivity (or flow) $\psi(S,v)$ from $S$ to any node $v$ is at least the demand $d_v$ of $v$.
In many {\sf SL} problems $\psi(S,v)=d_v$ if $v \in S$, namely,
the demand of nodes selected to $S$ is completely satisfied.
In a node-connectivity variant suggested recently by Fukunaga \cite{Fuk},
every node $v$ gets a ``bonus'' $p_v \leq d_v$ if it is selected to $S$,
namely, $\psi(S,v)=p_v+\kappa(S \setminus \{v\},v)$ if $v \in S$ and $\psi(S,v)=\kappa(S,v)$ otherwise,
where $\kappa(S,v)$ is the maximum number of internally disjoint $(S,v)$-paths.
While the approximability of many {\sf SL} problems was seemingly settled to $\Theta(\ln d(V))$ in \cite{SMF},
Fukunaga \cite{Fuk} showed that for undirected graphs one can achieve ratio $O(k \ln k)$ for his variant,
where $k=\max_{v \in V}d_v$ is the maximum demand.
We improve this by achieving ratio $\min\{p^* \ln k,k\} \cdot O(\ln (k/q^*))$
for a more general version with node capacities, where 
$p^*=\max_{v \in V} p_v$ is the maximum bonus and
$q^*=\min_{v \in V} q_v$ is the minimum capacity.
In particular, for the most natu\-ral case $p^*=1$ considered in \cite{Fuk} 
we improve the ratio from $O(k \ln k)$ to $O(\ln^2k)$.
We also get ratio $O(k)$ for the edge-connectivity version,
for which no ratio that depends on $k$ only was known before.
To derive these results, we consider a particular case of the 
{\sf Survivable Network} ({\sf SN}) problem when all edges of positive cost form a star.
We give ratio $O(\min\{\ln n,\ln^2 k\})$ for this variant, 
improving over the best ratio known for the general case $O(k^3 \ln n)$ 
of Chuzhoy and Khanna \cite{CK-new}.

In addition, we show that directed {\sf SL} 
with unit costs is $\Omega(\log n)$-hard to approximate even for $0,1$ demands,
while {\sf SL} with uniform demands can be solved in polynomial time. 
Finally, we consider a generalization of {\sf SL} where we also have 
edge-costs $\{c_e:e \in E\}$ and flow-cost bounds $\{b_v: v \in V\}$, and require that 
for every node $v$, the minimum cost of a flow of value $d_v$ from $S$ to $v$ is at most $b_v$.
We show that this problem admits approximation ratio $O(\ln d(V)+\ln(nc(E)-b(V))$.
\end{abstract}

\section{Introduction}

In {\sf Source Location} ({\sf SL}) problems,
the goal is to select a minimum cost source set $S \subseteq V$
such that the connectivity from $S$ to any node $v$ is at least the demand $d_v$ of $v$. 
Formally, the generic version of this problem is as follows.

\begin{center} 
\fbox{
\begin{minipage}{0.960\textwidth}
\noindent
{\sf Source Location} ({\sf SL}) \\
{\em Instance:}
A graph $G=(V,E)$ with node-costs $c=\{c_v:v \in V\}$, 
connectivity demands $d=\{d_v:v \in V\}$, and 
connectivity function $\psi:2^V \times V \rightarrow \mathbb{Z}_+$. \\
{\em Objective:}
Find a minimum cost source node set $S \subseteq V$ such that 
$\psi(S,v) \geq d_v$ for every $v \in V$.
\end{minipage}
}
\end{center}

Several connectivity functions $\psi$ appear in the literature.
To avoid considering many cases, we suggest two generic types, that include previous particular cases.

\begin{definition}
An integer set-function $f$ on a groundset $U$ is submodular if $f(X)+f(Y) \geq f(X \cap Y)+f(X \cup Y)$
for all $X,Y \subseteq U$, 
and $f$ is non-decreasing if $f(X) \leq f(Y)$ for all $X \subseteq Y \subseteq U$.
\end{definition}

\begin{definition} \label{d:pq}
Let $G=(V,E)$ be a graph with node-capacities $\{q_u:u \in V\}$.
For $S \subseteq V$ and $v \in V$ the 
{\em $(S,v)$-$q$-connectivity} $\lambda_G^q(S,v)$ 
is the maximum number of edge-disjoint paths from $S \setminus \{v\}$ to $v$ in $G$ 
such that every node $u$ is an internal node in at most $q_u$ paths.
Given {\em connectivity bonuses} $\{p_u \geq q_u:u \in V\}$, 
the {\em $(S,v)$-$(p,q)$-connectivity} $\lambda_G^{p,q}(S,v)$ is defined by: 
$\lambda_G^{p,q}(S,v)=p_v+\lambda_G^q(S,v)$ if $v \in S$, and
$\lambda_G^{p,q}(S,v)=\lambda_G^q(S,v)$ otherwise.
\end{definition}

We will say that a connectivity function $\psi(S,v)$ is submodular if for every $v \in V$,
the function $f_v(S)=\psi(S,v)$ is submodular and non-decreasing.
We will say that $\psi(S,v)$ is survivable if it is of the type $\psi(S,v)=\lambda_G^{p,q}(S,v)$.
It is not hard to see that every survivable connectivity function is submodular 
(see Section~\ref{s:sur}), but the inverse is not true in general. 
This gives only two types of {\sf SL} problems. 

\begin{center} 
\fbox{
\begin{minipage}{0.960\textwidth}
\noindent
{\sf Submodular SL}:     The connectivity function $\psi(S,v)$ is submodular. \\
{\sf Survivable SL}: \ \ The connectivity function $\psi(S,v)$ is survivable. 
\end{minipage}
}
\end{center}

We list four connectivity functions that appear in the literature.
All of them are submodular, and three of them are also survivable.
Given an {\sf SL} instance let 
$k=\max_{v \in V}d_v$ denote the maximum demand,
and in the case of {\sf Survivable SL} let 
$p^*=\max_{u \in V} p_u$ denote the maximum connectivity bonus and 
$q^*=\min_{u \in V} q_u$ denote the minimum node capacity.
In what follows assume that $1 \leq q_u \leq p_u \leq k$ for all $u \in V$,
and thus $1 \leq p^* \leq k$ and $1 \leq q^* \leq k$ holds.
\begin{enumerate}
\item
{\sf $\lambda$-SL}: 
$\lambda_G(S,v)$ is the maximum number of pairwise edge-disjoint $(S,v)$-paths if $v \notin S$ and 
$\lambda_G(S,v)=\infty$ otherwise. \\
This is {\sf Survivable SL} with $p_u=q_u=k$ for every $u \in V$.
\item
{\sf $\kappa$-SL}: 
$\kappa(S,v)$ is the maximum number of $(S,v)$-paths no two of which 
have a common node in $V \setminus (S \cup v)$ if $v \notin S$,
and $\kappa(S,v)= \infty$ otherwise. 
% For directed graphs, {\sf $\kappa$-SL} is equivalent to {\sf $\lambda$-SL}, 
% but we do not see that this connectivity function is survivable.
\item
{\sf $\hat{\kappa}$-SL}:
$\hat{\kappa}(S,v)$ is the maximum number of $(S,v)$-paths no two of which 
have a common node in $V \setminus \{v\}$ if $v \notin S$, and $\hat{\kappa}(S,v)= \infty$ otherwise. \\
This is {\sf Survivable SL} with $p_u=k$ and $q_u=1$ for every $u \in V$.
\item
{\sf $\kappa'$-SL}: 
$\kappa'(S,v)=\hat{\kappa}(S,v)$ if $v \notin S$ and 
$\kappa'(S,v)=p_v+\hat{\kappa}(S \setminus \{v\},v)$ if $v \in S$. \\
This is {\sf Survivable SL} with $q_u=1$ for every $u \in V$.
\end{enumerate}

\begin{table}[htbp] \label{tbl:ratios}
\begin{tabular}{|l||l|l||l|l|} \hline
{\hphantom{a} \boldmath $c$ \ \& \ $d$}  
                 & \multicolumn{2}{c||}{{\boldmath $\lambda$} ($p,q \equiv k$)}
                 & \multicolumn{2}{c|}{{\boldmath  $\kappa$}} \\\hline
                 & {\em Undirected} & {\em Directed}
                 & {\em Undirected} & {\em Directed}

\\\hline
GC \& GD         & $\Theta(\ln d(V))$ \cite{BKP,SMF}
                 & $\Theta(\ln d(V))$ \cite{BKP,SMF} 
                 & $\Theta(\ln d(V))$ \cite{BKP,SMF}
                 & $\Theta(\ln d(V))$ \cite{BKP,SMF}                  
\\\hline
GC \& UD         & in P \cite{AIMF}
                 & $O(\ln d(V))$ \cite{BKP}  
                 & $O(\ln d(V))$ \cite{BKP}  
                 & $O(\ln d(V))$ \cite{BKP}  
\\\hline
UC \& GD         & in P \cite{AIMF}
                 & $O(\ln d(V))$ \cite{BKP}  
                 & $O(\ln d(V))$ \cite{BKP}  
                 & $O(\ln d(V))$ \cite{BKP}  
\\\hline
UC \& UD         & in P \cite{TSSA}
                 & in P \cite{IMAH}
                 & $O(\ln d(V))$ \cite{BKP}  
                 & $O(\ln d(V))$ \cite{BKP}  
\\\hline \hline  
                 & \multicolumn{2}{c||}{{\boldmath $\hat{\kappa}$} ($p \equiv k$, $q\equiv 1$)}
                 & \multicolumn{2}{c|}{{\boldmath  $\kappa'$}      ($q \equiv 1$)} \\\hline
                 
GC \& GD         & $\Theta(\ln d(V))$ \cite{SMF}  
                 & $\Theta(\ln d(V))$ \cite{SMF}   
                 & $O(\ln d(V))$ \cite{Fuk}  
                 & $O(\ln d(V))$ \cite{Fuk}        
\\
                 & $O(k \ln k)$ \cite{Fuk}
                 &  
                 & $O(k \ln k)$ \cite{Fuk}
                 &       
\\\hline
GC \& UD         & in P \cite{NII}
                 & in P \cite{NII}
                 & % in P 
                 & % in P 
\\\hline
UC \& GD         & $O(\ln d(V))$ \cite{SMF}  
                 & $O(\ln d(V))$ \cite{SMF}  
                 & % $O(\ln d(V))$ \cite{Fuk}
                 & % $O(\ln d(V))$ \cite{Fuk}
\\
                 & $O(k)$ \cite{Ishii}
                 &
                 &
                 & % $\Omega(\log n)$ 
\\\hline
UC \& UD         & in P \cite{NII}
                 & in P \cite{NII}
                 & % in P 
                 & % in P 
\\\hline
\end{tabular}
\vspace{0.2cm}
\caption{
Previous approximation ratios and lower bounds for {\sf SL} problems. 
GC and UC stand for general and uniform costs, 
GD and UD stand for general and uniform demands, respectively.}
\vspace*{-0.7cm}
\end{table}

The known approximability status of {\sf SL} problems with connectivity functions 
$\lambda,\kappa,\hat{\kappa},\kappa'$, is summarized in Table~1; see also a survey in \cite{NI}.
The approximability of {\sf $\lambda,\kappa,\hat{\kappa}$-SL} problems was settled to $O(\ln d(V))$ in \cite{SMF}
(where $d(V)=\sum_{v \in V} d_v$),
while Fukunaga~\cite{Fuk} showed that undirected {\sf $\kappa'$-SL} admits ratio $O(k\ln k)$.
% Here we show that this is so for {\em all} {\sf Submodular SL} problems, 
% but for {\em undirected} {\sf Survivable SL} problems we can often achieve a better ratio.
We prove the following.

\begin{theorem} \label{t:SL}
{\sf Submodular SL} admits ratio $O(\ln d(V))$.
% {\sf Survivable SL} admits ratio $\ln\left(\min\{p^*n,d(V)\right)$.
Undirected {\sf Survivable SL} admits ratio $\min\{p^* \ln k,k\} \cdot O(\ln (k/q^*))$.
\end{theorem}

Theorem~\ref{t:SL} has several consequences.
While ratio $O(\ln (d(V))$ was known for connectivity functions 
$\lambda$,$\kappa,\hat{\kappa}$ \cite{SMF},
our proof of a more general result is simpler and shorter than the proof of each particular case.
For undirected graphs, the second part of Theorem~\ref{t:SL} implies that {\sf Survivable SL} problems admit 
ratio $O(k \ln (k/q^*))$ if $p^* \geq k/\ln k$ (e.g., $p^*=k$ in {\sf $\lambda$-SL} and {\sf $\hat{\kappa}$-SL}),
and ratio $O(p^* \ln k \ln (k/q^*))$ if $p^* < k/\ln k$ (e.g., {\sf $\kappa'$-SL} with $p^*=1$).
% This improves and generalizes the ratio $O(k \ln k)$ for {\sf $\kappa'$-SL} of \cite{Fuk}.
In the case of $\lambda$-{\sf SL} we have $q^*=k$ which implies ratio $O(k)$ -- 
this is the first ratio for $\lambda$-{\sf SL} that depends on $k$ only.
Summarizing, we have the following new results for connectivity functions $\lambda,\kappa'$. 

\begin{corollary}
{\sf $\lambda$-SL} admits ratio $O(k)$ and 
{\sf $\kappa'$-SL} admits ratio $O(p^* \ln^2 k)$.
\end{corollary}

To prove Theorem~\ref{t:SL}, we consider % a (node-capacitated version of) 
the following known problem.

\begin{center} 
\fbox{
\begin{minipage}{0.960\textwidth}
\noindent
{\sf Survivable Network} ({\sf SN}) \\
{\em Instance:} \ 
A graph $G=(V,E)$ with edge-costs $\{c_e:e \in E\}$ and node capacities $\{q_u:u \in V\}$,
and connectivity requirements $r=\{r_{sv}:sv \in D\}$ on a set $D$ of demand edges on $V$. \\
{\em Objective:}
Find a minimum-cost subgraph $G'$ of $G$ such that  
$\lambda^q_{G'}(s,v) \geq r_{sv}$ for every $sv \in D$.
\end{minipage}
}
\end{center}

Let $k=\max_{sv \in D} r_{sv}$ denote the maximum requirement.
For $q \equiv k$ we get the edge-connectivity version (which admits ratio $2$ \cite{Jain}), 
while for $q \equiv 1$ we get the node-connectivity version.
{\sf SN} admits a folklore ratio $O(|D|)$, 
and for directed graphs no better ratio is known.
Undirected {\sf SN} admits ratios $O\left(k^3 \log n \right)$ \cite{CK-new},
% for edge-costs and $O\left(k^4 \log^2 n\right)$ for node-costs \cite{N-Focs,Vakilian},
and has an $\Omega(\{k^{1/4},|D|^{1/6})$ approximation lower bound \cite{Laek}.  
We consider the following particular case of {\sf SN}, studied previously in \cite{KN-aug,Fuk}.

\begin{center} 
\fbox{
\begin{minipage}{0.960\textwidth}
{\ssn}: the set $F$ of edges in $E$ of positive cost is a star with center $a$. 
\end{minipage}
}
\end{center}

The {\ssn} problem was defined in \cite{KN-aug},
where it was shown to admit ratio $O(\ln n)$ for unit edge-costs.
The study of this problem in \cite{KN-aug} is motivated by the observation that directed {\sf SN} instances 
when $(V,F)$ is a complete graph with unit edge costs (so called {\sf Connectivity Augmentation} problem)
can be reduced to {\sf Star-NA} with a loss of a factor of $2$ in the approximation ratio. 
Fukunaga \cite{Fuk} observed that {\sf $\kappa'$-SL} is a special case of {\sf Star-SN}. 
Hence {\sf Star-SN} problems are important, as they generalize several well known problems.

Our results for {\sf Star-SN}, summarized in the following theorem, 
substantially improve over the best known ratios for {\sf SN}. 
These results are of independent interest, as they show that {\sf Star-SN} 
admits much better ratios than general {\sf SN}.

\begin{theorem} \label{t:a}
{\sf Star-SN} admits approximation ratios $O(\ln n)$ for directed graphs, 
and $O(\min\{\ln n,\ln k \ln(k/q^*)\})$ for undirected graphs.
% In the case of node-costs, the problem admits approximation ratios  for directed graphs and 
% $\min\{p^* \ln k,k\} \cdot O(\ln (k/q))$ for undirected graphs, 
% where $p^*$ is the maximum number of parallel edges in $F$. 
\end{theorem}

We further study {\sf SL} problems and prove the following.

\begin{theorem} \label{t:hard}
Directed {\sf Survivable SL} for $k=1$ and unit costs is $\Omega(\log n)$-hard to approximate.
Directed/undirected {\sf $\kappa'$-SL} with uniform demands and with $p \equiv 1$ 
can be solved in polynomial time. 
\end{theorem}

Finally, we consider the following generalization of {\sf Survivable SL}. 
Given an instance of {\sf Survivable SL} and edge-costs $c=\{c_e:e \in E\}$, 
let $\mu^{p,q}_G(S,v)$ denote the minimum cost of an edge set  
$F \subseteq E$ such that $\lambda^{p,q}_{(V,F)}(S,v) \geq d_v$,
where $\mu_G^{p,q}(S,v)=\infty$ if no such edge set $F$ exists (namely, if $\lambda_G^{p,q}(S,v) < d_v$).

\begin{center} 
\fbox{
\begin{minipage}{0.960\textwidth}
\noindent
{\sf Survivable SL with Flow-Cost Bounds}  \\
{\em Instance:}
As in {\sf Survivable SL}, but in addition we are also given edge-costs $\{c_e:e \in E\}$
and flow-cost bounds $\{b_v \leq c(E): v \in V\}$. \\
{\em Objective:}
As in {\sf Survivable SL}, with an additional constraint $\mu_G^{p,q}(S,v) \leq b_v$ for every $v \in V$.
\end{minipage}
}
\end{center}

\begin{theorem} \label{t:fc-bounds}
{\sf Survivable SL with Flow-Cost Bounds} admits approximation ratio $H(d(V))+H(nc(E)-b(V))$.
\end{theorem}

% Theorems \ref{t:hard}, \ref{t:aug}, \ref{t:a}, \ref{t:main'}, and \ref{t:fc-bounds}, are proved in
% Sections \ref{s:hard}, \ref{s:aug}, \ref{s:a}, \ref{s:main'}, and \ref{s:fc-bounds}, respectively.

% \newpage

\section{Relations between {\sf SL} and {\sf SN} problems}

In this section we explain the relation between {\sf SL} and {\sf SN} problems. 
For that, it would be convenient to consider the augmentation version of the {\sf SN} problem,
with arbitrary connectivity functions and allowing node-costs. 
Given a function $w=\{w_u:u \in U\}$ on a groundset $U$ and $U' \subseteq U$, 
let $w(U')=\sum_{u \in U'}w_u$. If $w$ is a cost function on $U$ and $I$ is an edge-set on $U$,
then the cost (or the node-costs) $w(I)$ of $I$ is the cost of the set of the endnodes of $I$.
Formally, we call the problem we need {\sf Network Augmentation}, and it is defined as follows.

\begin{center} 
\fbox{
\begin{minipage}{0.960\textwidth}
\noindent
{\sf Network Augmentation} ({\sf NA}) \\
{\em Input:} \ 
A graph $G=(V,E)$, an edge-set $F$ on $V$, a cost function $c$ on $F$ or on $V$, 
connectivity requirements $r=\{r_{sv}:sv \in D\}$ on a set $D$ of demand edges on $V$,
and a family $\{f_{sv}:2^F\rightarrow \mathbb{Z}_+:sv \in D\}$ of connectivity functions. \\
{\em Output:}
A min-cost edge-set $I \subseteq F$ such that  
$f_{sv}(I) \geq r_{sv}$ for every $sv \in D$.
\end{minipage}
}
\end{center}

As in the case of {\sf SL} problems, we consider two types of {\sf NA} problems:

\begin{center} 
\fbox{
\begin{minipage}{0.960\textwidth}
{\sf Submodular NA}: connectivity functions are submodular and non-decreasing. \\
{\sf Survivable NA}: \ \ connectivity functions are $f_{sv}(I)=\lambda^q_{G+I}(s,v)$. 
\end{minipage}
}
\end{center}

Note that {\sf SN} is a particular case of {\sf Survivable NA}
when $E=\emptyset$, but it is easy to see that for edge-costs the problems are equivalent.
% but we do not have a proof that this is so for node-costs. 
As we shall see, {\sf SL} is related to the following two particular cases of {\sf NA} with node-costs:

\begin{center} 
\fbox{
\begin{minipage}{0.960\textwidth}
{\sf Rooted NA}:     \ \ $D$ is a star with center $s$. \\
{\cna}:                  $D,F$ are both stars with a common center $s$.
\end{minipage}
}
\end{center}

Fukunaga \cite{Fuk} observed that {\sf $\kappa'$-SL} is equivalent 
(via an approximation ratio preserving reduction) to 
{\sf Survivable {\cna}} with {\em edge-costs} and $q \equiv 1$. 
Here we further observe the following.
For an edge-set/graph $J$ let $\delta_J(X)$ denote the set of edges in $J$ from $X$ to $V \setminus X$. 

\begin{observation} \label{o:reduction}
For both directed and undirected graphs, 
{\sf Survivable SL} is equivalent to {\sf Survivable {\cna}} 
with node-costs such that $\delta_G(s)=\emptyset$ and $c(s)=0$.
\end{observation}
\begin{proof}
Given an instance of {\sf Survivable SL} construct an instance of {\sf Survivable {\cna}} as follows:
add to $G$ a new node $s$ of cost $0$, and for every $v \in V$ set $r_{sv}=d_v$ and 
put $p_v$ edges from $s$ to $v$ into $F$.
Conversely, given an instance of {\sf Survivable {\cna}}, 
construct an instance of {\sf Survivable SL} as follows. 
Remove $s$ from $G$, and for every $v \in V$ set 
$p_v$ to be the number of edges in $F$ from $s$ to $v$ and $d_v=r_{sv}$.
In both directions, it is not hard to see that $S$ is a solution to the {\sf Survivable SL} instance,
if, and only if, the edge set $I$ of all edges in $F$ from $s$ to $S$  
is a solution to the {\sf Survivable {\cna}} instance, and clearly $I$ and $S$ have the same node-cost.  
\qed
\end{proof}

The best known ratios for {\sf Survivable NA} are $O(k^3 \ln n)$ for edge-costs \cite{CK-new},
and $O\left(k^4 \log^2 |D|\right)$ for node-costs \cite{N-Focs,Vakilian}.
The best known ratio for {\sf Rooted Survivable NA} are $O(k \ln k)$ for edge-costs \cite{N-Focs}
and $O\left(k^2 \log n\right)$ for node-costs \cite{N-Focs,Vakilian},
and no better ratios were known even for {\sf Survivable {\cna}}, see \cite{Fuk} where 
ratio $O(k\ln k)$ for undirected {\sf Survivable {\cna}} was deduced in two ways: 
from the ratio $O(k\ln k)$ for {\sf Rooted Survivable NA} \cite{N-Focs}, and via iterative rounding.
Our results for {\sna}, that imply Theorems \ref{t:SL} and \ref{t:a}, 
are summarized in the following three theorems. Let $H(j)$ denote the $j$th Harmonic number,
and here for an {\sf NA} instance let $p^*$ denote the maximum number of parallel edges in $F$.

\begin{theorem} \label{t:sub}
Directed {\sf Submodular {\sna}} admits ratio $H(\alpha)$ in the case of edge-costs and ratio 
$H(\beta)$ in the case of node-costs, where 
$$
\alpha = \max\limits_{e \in F} \sum_{uv \in D} [f_{uv}(\{e\})-f_{uv}(\emptyset)] \ \ \ \ \ 
\beta = \max\limits_{z \in V} \sum_{uv \in D}[f_{uv}(\delta_F(z))-f_{uv}(\emptyset)] \ .$$
\end{theorem}

\begin{theorem} \label{t:sur}
For directed graphs, {\sf Survivable NA} is a particular case of {\sf Submodular NA},
and $\alpha \leq |D|$ and $\beta \leq \min\{r(D),p^*|D|\}$ holds for {\sf Survivable Star-NA}.
\end{theorem}

\begin{theorem} \label{t:und}
Undirected {\sf Survivable {\sna}} admits ratio $O(\ln k \ln (k/q^*)))$ for edge-costs and 
$\min\{p^* \ln k,k\} \cdot O(\ln (k/q^*))$ for node-costs.
\end{theorem}

We briefly mention the techniques we use to prove these theorems.
Theorem~\ref{t:sub} is essentially an easy application of the greedy algorithm of Wolsey \cite{W}
for the {\sf Submodular Cover} problem.
Parts of Theorem~\ref{t:sur} were implicitly proved in \cite{KN-aug}, but our proof 
is both more general and substantially simpler.
Our main technical contribution is Theorem~\ref{t:und}.
To prove this theorem, we consider the augmentation version of {\sf Survivable {\sna}} with edge-costs 
where the goal is to increase the connectivity by one between the pairs in $D$.
Using LP-scaling we show that ratio $\rho$ for the augmentation version implies ratio 
$O(\rho \ln k)$ for the edge-costs version of the general problems, and ratio 
$\min\{p^* \ln k, k\} \cdot O(\rho)$ for the node-costs version.
Then we design an $O(\ln (k/q^*))$-approximation algorithm for the augmentation version.
This is achieved by formulating the augmentation problem as a {\sf Biset-Family Edge-Cover} problem, 
reducing the later problem to the problem of finding a minimum cost vertex cover in a hypergraph,
and using a theorem from \cite{N-aug} to show that the maximum degree in the obtained hypergraph 
is $O\left({(k/q^*)}^2\right)$.

\section{Proof of Theorem~\ref{t:sub}} \label{s:sub}

All graphs in this and the next sections are assumed to be directed. 
To prove Theorem~\ref{t:sub} we use a result due to Wolsey \cite{W} about a performance of a 
greedy algorithm for submodular covering problems. In a generic covering problem we are given
by a value oracle two set functions on a groundset $U$: 
a cost-function $c:2^U \rightarrow \mathbb{R}$ and a progress function $g:2^U \rightarrow \mathbb{Z}$.  
The goal is to find $A \subseteq U$ of minimum cost such that $g(A)=g(U)$.
The {\sf Submodular Cover} problem is a special case when the function $g$ is submodular and non-decreasing,
and $c(S)=\sum_{v \in S} c(v)$ for some $c:U\rightarrow \mathbb{R}^+$.
Wolsey \cite{W} proved that then, the greedy algorithm, that 
% starts with $A=\emptyset$ and 
as long as $g(A)<g(U)$ repea\-tedely adds to $A$ an element $u \in U \setminus A$ with maximum
$\frac{g(A \cup \{u\})-g(A)}{c_u}$, has approximation ratio 
$H\left(\max_{u \in U} g(\{u\})-g(\emptyset)\right)$.

We start with the case of edge-costs. Then 
the function $g$ is defined in the same way as in \cite{KN-aug,SMF}:
$U=F$ and for $I \subseteq F$ 
$$g(I)=\sum_{uv \in D}\min\{r(u,v),f_{uv}(I)\}  .$$
It is not hard to verify that $g$ is non-decreasing, 
and that $I$ is a feasible solution to an {\sf NA} instance if and only if $g(I)=g(F)=r(D)$.
Also, $g(\{e\}) - g(\emptyset) \leq \sum_{uv \in D} [f_{uv}(\{e\}) - f_{uv}(\emptyset)]$ for any $e \in F$.
% Also, for any $e \in F$, $f_{uv}(\{e\}) \leq  f_{uv}(\emptyset)+1$, hence 
We show that $g$ is submodular.
It is known (c.f. \cite{Sch}) that if $h$ is submodular, then $\min\{r,h\}$ is submodular
for any constant $r$.   
Thus the function $h_{uv}(I)=\min\{r(u,v),f_{uv}(I)\}$ is submodular.
As a sum of submodular functions is also submodular, we obtain that $g$ is submodular.

Now let us consider node-costs. 
For $S \subseteq V$ let $F_S$ denote the set of edges in $F$ from $a$ to $S$, and let $f'_{uv}(S)=f_{uv}(F_S)$.
We have $U=V$ and for $S \subseteq V$ let 
$$g'(S)= \sum_{uv \in D}\min\{r(u,v),f'_{uv}(S)\} \ .$$
As in the edge-costs case, it is not hard to verify that $g'$ is non-decreasing and that 
$S$ is a feasible solution to an {\sf NA} instance if and only if $g'(S)=g'(V)=r(D)$.
% Also, for any $z \in V$, $f'_{uv}(\{z\})-f_{uv}(\emptyset) \leq |\delta_F(z)| \leq p_{\max}$, hence 
Also, $g'(\{z\}) - g'(\emptyset) \leq \sum_{uv \in D} [f_{uv}(\delta_F(z)) - f_{uv}(\emptyset)]$
for any $z \in V$.
We show that $g'$ is submodular. We claim that the submodularity of $f(I)$ implies 
that $f'(S)$ is submodular. This is not true in general, but holds if $F$ is a star, 
and hence for {\sna} instances.
More precisely, it is not hard to verify the following statement, that 
finishes the proof of Theorem~\ref{t:sub}.

\begin{lemma} \label{l:star}
Let $(V,F)$ be a graph and let $f$ be a submodular set function on $F$.
If $F$ is a star, then the set function $f'(S)=f(F_S)$ defined on $V$ is also submodular.
\end{lemma}

\section{Proof of Theorem~\ref{t:sur}} \label{s:sur}

We start by showing that in the case of edge-costs,
directed {\sf Survivable NA} is a particular case of {\sf Submodular NA}.
Let $s,v \in V$.
It is easy to see that $f(I)=f_{sv}(I)=\lambda^q_{G+I}(s,v)$ is non-decreasing.
We prove that $f(I)$ is submodular.
For that, we will use the following known characterization of submodularity, c.f. \cite{Sch}: \\
{\em A set-function $f$ on $F$ is submodular if, and only if}
$$
f(I_0 \cup \{e\})+f(I_0 \cup \{e'\}) \geq f(I_0)+f(I_0 \cup \{e,e'\})  \ \ \ \forall I_0 \subset F, e,e' \in F \setminus I_0 \ 
$$
Let us fix $I_0 \subseteq F$.
Revising our notation to $G \gets G+I_0$, $F \gets F \setminus I_0$,
and denoting $h(I)=f(I_0 \cup I)-f(I_0)$, we get that $f$ is submodular if, and only if
$$
h(\{e\})+h(\{e'\}) \geq h(\{e,e'\})  \ \ \ \forall e,e' \in F  \ . 
$$
In our setting, $h(I)=\lambda^q_{G+I}(s,v)-\lambda^q_G(s,v)$ is the increase in 
the $(s,v)$-$q$-connectivity as a result of adding $I$ to $G$.
Thus $0 \leq h(I) \leq |I|$ for any $I \subseteq F$, so 
$0 \leq h(\{e,e'\}) \leq 2$. 
If $h(\{e,e'\})=0$, then we are done;
if $h(\{e,e'\})=1$, then we need to show that $h(\{e\})=1$ or $h(\{e'\})=1$; and 
if $h(\{e,e'\})=2$, then we need to show that $h(\{e\})=1$ and $h(\{e'\})=1$.
We prove the following general statement, that implies the above.

\begin{lemma} \label{l:J}
Let $G=(V,E)$ be a directed graph with node capacities $\{q_v:v \in V\}$,
let $I$ be a set of edges on $V$ disjoint to $E$, let $s,t \in V$, 
and let $h=\lambda^q_{G+I}(s,t)-\lambda^q_G(s,t)$.
Then there is $J \subseteq I$ of size $|J| \geq h$ 
such that $\lambda^q_{G+\{e\}}(s,t)=\lambda^q_G(s,t)+1$ for every $e \in J$.
\end{lemma}
\begin{proof}
Since we consider directed graphs, 
it is sufficient to prove the lemma for the case of edge-connectivity.
For that, apply the following standard reduction that eliminates node capacities:
replace every $v \in V \setminus \{s,t\}$ by two nodes $v^{in}, v^{out}$ connected by $q_v$ parallel edges from  
$v^{in}$ to $v^{out}$ and replace every $uv \in E \cup F$ by an edge from $u^{out}$ to $v^{in}$.
Hence we will prove the lemma for the edge connectivity function $\lambda$.
Let us say that $S \subseteq V$ is {\em tight} if $s \in S$, $v \notin S$, and 
$|\delta_G(S)|=\lambda_G(s,v)$. Let ${\cal F}$ be the family of tight sets.
By Menger's Theorem ${\cal F}$ is non-empty.
It is known that ${\cal F}$ is a ring family, namely, the intersection of all the sets in ${\cal F}$ is 
nonempty, and if $X,Y \in {\cal F}$ then $X \cap Y, X \cup Y \in {\cal F}$. Then ${\cal F}$ has 
a unique inclusion-minimal set $S_{\min}$ and 
a unique inclusion-maximal set $S_{\max}$.
Let $J=\{uv \in I:u \in S_{\min},v \in V \setminus S_{\max}\}$ 
be the set of edges in $I$ that go from $S_{\min}$ to $V \setminus S_{\max}$.
By Menger's Theorem, $|J| \geq h$, and $\lambda_{G+\{e\}}(s,t)=\lambda_G(s,t)+1$
for any $e \in J$. The statement follows.
\qed
\end{proof}

The bound $\beta \leq r(D)$ is obvious, while the other bounds
on $\alpha$ and $\beta$ follow from the simple observation that
for any $s,v \in V$, the set-function on $F$ defined by $f(I)=\lambda^q_{G+I}(s,v)$ 
has the following properties: 
$f(\{e\}) \leq 1$ for any $e \in F$ and $f(\delta_F(z)) \leq |\delta_F(z)| \leq p^*$ for any $z \in V$.

% \begin{lemma} \label{l:max}
% For any directed {\sf Star-NA} instance, for any $s,v \in V$, the set-function $f_{sv}(I)=\lambda^q_{G+I}(s,v)$ 
% on $F$ is submodular non-decreasing, $f_{sv}(\{e\}) \leq  f_{sv}(\emptyset)+1$ for any $e \in F$, and
% $f'_{sv}(\{z\})-f'_{sv}(\emptyset) \leq |\delta_F(z)| \leq p^*$ for any $z \in V$. 
% \end{lemma}

% The proof of Theorem~\ref{t:sur} is complete.

\section{Proof of Theorem~\ref{t:und}} \label{s:und}

% Here we prove Theorem~\ref{t:main'}.
All graphs in this and the next section are assumed to be undirected. % unless stated otherwise.
We start by considering the edge-costs case, and then will show that it implies the node-costs
case by reductions.

\begin{definition}
An ordered pair $\hat{X}=(X,X^+)$ of subsets of a groundset $V$ is called a {\em biset} if $X \subseteq X^+$; 
$X$ is the {\em inner part} and $X^+$ is the {\em outer part} of $\hat{X}$,
and $\Gamma(\hat{X})=X^+ \setminus X$ is the {\em boundary} of $\hat{X}$. 
An edge $e$ covers a biset $\hat{X}$ if it has one endnode in $X$ and the other in $V \setminus X^+$.
For a biset $\hat{X}$ and an edge-set/graph $J$ let $\delta_J(\hat{X})$ denote the 
set of edges in $J$ covering $\hat{X}$. 
\end{definition}

Given an instance of {\sf Survivable NA} and a biset $\hat{X}$ on $V$, 
let the requirement of $\hat{X}$ be $r(\hat{X})=\max\{r_{uv}: uv \in \delta_D(\hat{X})\}$
if $\delta_D(\hat{X}) \neq \emptyset$ and $r(\hat{X})=0$ otherwise.
By the $q$-connectivity version of Menger's Theorem (c.f. \cite{KN-sur}), 
$I \subseteq F$ is a feasible solution to an {\sf Survivable NA}
instance if, and only if, $|\delta_I(\hat{X})| \geq h(\hat{X})$ for every bisets $\hat{X}$ on $V$,
where $h$ is a biset-function defined by
\begin{equation} \label{e:p}
h(\hat{X})= \max\{r(\hat{X})-(q(\Gamma(\hat{X}))+|\delta_G(\hat{X})|),0\} 
\end{equation}
Let ${\cal P}_h$ denote the polytope of ``fractional edge-covers'' of $h$, namely,
$$
{\cal P}_h=\left\{x \in \mathbb{R}^F:
x\left(\delta_F(\hat{Y})\right) \geq h(\hat{Y}) \ \forall \mbox{ biset } \hat{Y} \mbox{ on } V, \
0 \leq x_e \leq 1 \ \forall e \in F\right\} \ .
$$
Let $\tau(h)$ denote the optimal value of a standard LP-relaxation 
for edge covering $h$ by a minimum cost edge set, namely, 
$\tau(h)=\min\left\{\sum_{e \in F} c_e x_e: x \in {\cal P}_h \right\}$.
% and $\tau(h)=\min\left\{\sum\limits_{v \in V} \left(c_v \max\limits_{e \in \delta_F(v)}x_e\right): 
% x \in {\cal P} \right\}$ in the case of node-costs.

As an intermediate problem, we consider {\sf Survivable NA} instances when 
we seek to increase the connectivity by $1$ for every $uv \in D$, namely, when 
$r_{uv} = \lambda^q_G(u,v)+1$ for all $uv \in D$. 

\begin{center} 
\fbox{
\begin{minipage}{0.960\textwidth}
\noindent
{\sf $D$-Survivable NA} (the edge-costs version) \\
{\em Input:} \ 
A graph $G=(V,E)$ with node-capacities $\{q_v:v \in V\}$, an edge set $F$ on $V$, 
a cost function $c$ on $F$, and a set $D$ of demand edges on $V$. \\
{\em Output:}
Find a min-cost edge-set $I \subseteq E$ such that 
$\lambda^q_{G+I}(u,v) \geq \lambda^q_G(u,v)+1$ for all $uv \in D$.
\end{minipage}
}
\end{center}

Given a {\sf $D$-Survivable NA} instance, we say that a biset $\hat{X}$ is {\em tight} if 
$h(\hat{X})=1$, where $h$ is defined by (\ref{e:p}).
{\sf $D$-Survivable NA} is equivalent to the problem of finding a minimum cost edge-cover of the biset family 
${\cal F}=\{\hat{X}:h(\hat{X})=1\}$ of tight bisets.
Thus the following generic problem includes the {\sf $D$-Survivable NA} problem.

\begin{center} 
\fbox{
\begin{minipage}{0.960\textwidth}
\noindent
{\sf Biset-Family Edge-Cover} \\
{\em Input:} \   
A graph $(V,F)$ with edge-costs and a biset family ${\cal F}$ on $V$. \\
{\em Output:}
Find a min-cost ${\cal F}$-cover $I \subseteq F$. 
\end{minipage}
}
\end{center}

For a biset-family ${\cal F}$ let $\tau({\cal F})$ denote the optimal value of a standard LP-relaxation 
for edge covering ${\cal F}$ by a minimum cost edge set, namely,
$\tau({\cal F})=\tau(h)$ for $h(\hat{X})=1$ if $\hat{X} \in {\cal F}$ and $h(\hat{X})=0$ otherwise.
% The following standard statement is known; we provide a proof for completeness of exposition. 

\begin{proposition} \label{p:scale}
Suppose that {\sf $D$-Survivable Star-NA} with edge-costs admits a polynomial time algorithm 
that computes a solution of cost at most $\rho(k) \tau({\cal F})$, where ${\cal F}$ is the family of tight bisets.
Then {\sf Survivable Star-NA} admits a polynomial time algorithm that computes a solution $I$ 
such that: 
\begin{itemize}
\item
For edge-costs,
$c(I) \leq \tau(h) \cdot \sum_{\ell=1}^k \frac{\rho(\ell)}{k-\ell+1}$, where $h$ is defined by (\ref{e:p}).
\item
For node-costs,
$c(I) \leq {\sf opt} \cdot \sum_{\ell=1}^k \rho(\ell) \cdot \min\left\{\frac{p^*}{k-\ell+1},1\right\}$.
\end{itemize}
\end{proposition}
\begin{proof}
We start with the edge-costs case. 
Consider the following sequential algorithm. Start with $I=\emptyset$.
At ite\-ra\-tion $\ell=1, \ldots,k$, add to $I$ and remove from $F$ an edge-set $I_\ell \subseteq F$ 
that increases by $1$ the $q$-connectivity of $G+I$ on the set of demand edges
$D_\ell=\{sv:\lambda_{G+I}^q(s,v)=r(s,v)-k+\ell-1, sv \in D\}$, 
by covering the corresponding biset-family ${\cal F}_\ell$ using the $\rho$-approximation algorithm.
After ite\-ration $\ell$, we have $\lambda_{G+I}^q(s,v) \geq r(s,v)-k+\ell$ for all $sv \in D$.
Consequently, after $k$ iterations $\lambda_{G+I}^q(s,v) \geq r(s,v)$ holds for all $sv \in D$,
thus the computed solution is feasible. 
The approximation ratio follows from the following two observations. 
\begin{itemize}
\item[(i)]
$c(I_\ell) \leq \rho(\ell) \cdot \tau({\cal F}_\ell)$.
This is so since $\lambda(s,v) \leq \ell-1$ for every $sv \in D_\ell$, 
hence the maximum requirement at iteration $\ell$ is at most $\ell$. 
\item[(ii)]
$\tau({\cal F}_\ell) \leq \frac{\tau(h)}{k-\ell+1}$.
To see this, note that if $\hat{Y} \in {\cal F}_\ell$ and $x \in {\cal P}_h$
then $x(\delta(\hat{Y})) \geq k-\ell+1$, by Menger's Theorem. 
Thus $x/(k-\ell+1)$ is a feasible solution for the LP-relaxation for edge-covering ${\cal F}_\ell$,
of value $c \cdot x/(k-\ell+1)$. 
\end{itemize}
Consequently, $c(I) = 
\sum_{\ell=1}^k c(I_\ell) \leq \sum_{\ell=1}^k \rho(\ell) \cdot \frac{\tau(h)}{k-\ell+1}=
\tau(h) \cdot \sum_{\ell=1}^k \frac{\rho(\ell)}{k-\ell+1}$.

Now let us consider the case of node-costs. 
Then we convert node-costs into edge-costs by assigning to every edge $e=av$ the cost $c'(e)=c(v)$. 
Let ${\sf opt}'$ denote the optimal solution value of the edge-costs instance obtained.
Clearly, ${\sf opt} \leq {\sf opt}' \leq p^* \cdot {\sf opt}$.
Note that any inclusion minimal solution to a {\sf $D$-Survivable NA} instance has no parallel edges.
This implies that $c(I_\ell) \leq \rho(\ell) \cdot {\sf opt}$ and that $c(I_\ell)=c'(I_\ell)$.
The latter implies 
$c(I_\ell)=c'(I_\ell) \leq \rho(\ell) \cdot \frac {{\sf opt}'}{k-\ell+1} \leq 
\rho(\ell) \cdot {\sf opt} \cdot \frac {p^*}{k-\ell+1}$,
and the statement for the node-costs case follows. 
\qed
\end{proof}

In the next section we prove the following theorem, 
that together with Proposition~\ref{p:scale} finishes the proof of Theorem~\ref{t:und}.

\begin{theorem} \label{t:D}
Undirected {\sf $D$-Survivable Star-NA} with edge-costs admits 
a polynomi\-al time algorithm that computes a solution $I$ of cost 
$\tau({\cal F}) \cdot O(\ln (k/q^*))$.
\end{theorem}

\section{Proof of Theorem~\ref{t:D}}

Recall that {\sf $D$-Survivable NA} reduces to {\sf Biset-Family Edge-Cover} with ${\cal F}$ 
being the family of tight bisets; in the case of rooted requirements, when $D$ is a star with center $s$,
it is sufficient to cover the biset-family ${\cal F}^s=\{\hat{X} \in {\cal F}: s \in V \setminus X^+\}$
Biset-families arising from {\sf Survivable NA} instances have some special properties,
that are summarized in the following definitions.

\begin{definition}
The intersection and the union of two bisets $\hat{X},\hat{Y}$ is defined by
$\hat{X} \cap \hat{Y} = (X \cap Y, X^+ \cap Y^+)$ and 
$\hat{X} \cup \hat{Y} = (X \cup Y,X^+ \cup Y^+)$. 
The biset $\hat{X} \setminus \hat{Y}$ is defined by 
$\hat{X} \setminus \hat{Y}=(X \setminus Y^+,X^+ \setminus Y)$.
We write $\hat{X} \subseteq \hat{Y}$ and say that 
{\em $\hat{Y}$ contains $\hat{X}$} if $X \subseteq Y$ and $X^+ \subseteq Y^+$.
Let ${\cal C}_{\cal F}$ denote the inclusion-minimal bisets in ${\cal F}$.
\end{definition}

\begin{definition} \label{d:uncrossable}
Two bisets $\hat{X},\hat{Y}$ covered by an edge-set $D$ are {\em $D$-independent}
if for any $xx',yy' \in D$ such that $xx'$ covers $\hat{X}$ and $yy'$ covers $\hat{Y}$,
$\{x,x'\} \cap \Gamma(\hat{Y}) \neq \emptyset$ or $\{y,y'\} \cap \Gamma(\hat{X}) \neq \emptyset$;
otherwise, $\hat{X},\hat{Y}$ are {\em $D$-dependent}.
We say that a biset family ${\cal F}$ is {\em $D$-uncrossable} if $D$ covers ${\cal F}$ and if
for any $D$-dependent $\hat{X},\hat{Y} \in {\cal F}$ the following holds:
\begin{equation} \label{e:uncross}
\hat{X} \cap \hat{Y},\hat{X} \cup \hat{Y} \in {\cal F} \mbox{ or }
\hat{X} \setminus \hat{Y},\hat{Y} \setminus \hat{X} \in {\cal F} \ .
\end{equation}
Similarly, given a set $T \subseteq V$ of terminals, we say that $\hat{X}, \hat{Y}$ are {\em $T$-independent}
if $X \cap T \subseteq \Gamma(\hat{Y})$ or if $Y \cap T \subseteq \Gamma(\hat{X})$, and 
$\hat{X},\hat{Y}$ are {\em $T$-dependent} otherwise.
We say that ${\cal F}$ is $T$-uncrossable if $T$ covers the set-family of the inner parts of ${\cal F}$,
and if (\ref{e:uncross}) 
% $\hat{X} \cap \hat{Y},\hat{X} \cup \hat{Y} \in {\cal F}$ or 
% $\hat{X} \setminus \hat{Y},\hat{Y} \setminus \hat{X} \in {\cal F}$
holds for any $T$-dependent $\hat{X},\hat{Y} \in {\cal F}$.
\end{definition}

A biset-family ${\cal F}$ is symmetric if $\hat{X} \in {\cal F}$ implies 
$(V \setminus X^+,V \setminus X) \in {\cal F}$.
% Clearly, the family of tight bisets is symmetric.
We will use the the following statement, that was implicitly proved in \cite{N-aug}.

\begin{lemma} [\cite{N-aug}]
The family ${\cal F}$ of tight bisets is symmetric and $D$-uncrossable; 
if $D$ is a star with % center $s$ and 
leaf-set $T$ then $\{\hat{X} \in {\cal F}: s \notin X^+\}$ is $T$-uncrossable. 
\end{lemma}

For a biset-family ${\cal C}$ let 
$\gamma_{\cal C}=\max\{|\Gamma(\hat{C})|: \hat{C} \in {\cal C}\}$.
Note that if ${\cal F}$ is the family of tight bisets then $\gamma_{\cal F} \leq (k-1)/q^*$.
Given an instance of {\sf Biset-Family Edge-Cover}, we will assume that the family ${\cal C}$ 
of the inclusion members of ${\cal F}$ can be computed in polynomial time.
We note that for ${\cal F}$ being the family of tight bisets,
this step can be implemented in polynomial time, c.f. \cite{N-aug}. 
Under this assumption, we prove the following generalization of Theorem~\ref{t:D}.

\begin{theorem} \label{t:main}
For edge/node-costs, {\sf Biset-Family Edge-Cover} with $F$ being a star admits 
a polynomial time algorithm that computes a cover $I$ of ${\cal F}$ such that:
\begin{itemize}
\item[{\em (i)}]
$c(I) \leq H\left({(4\gamma_{\cal C}+1)}^2\right) \cdot \tau({\cal F})$  
if ${\cal F}$ is symmetric and $D$-uncrossable.
\item[{\em (ii)}]
$c(I) \leq H(2\gamma_{\cal C}+1) \cdot \tau({\cal F})$ if ${\cal F}$ is $T$-uncrossable 
and $a \in V \setminus X^+$ for all~$\hat{X} \in {\cal F}$.
\end{itemize}
\end{theorem}

In the rest of this section we prove Theorem~\ref{t:main}. 

\begin{definition}
A node set $U \subseteq V$ is a {${\cal C}$-transversal} of a hypergraph (set-family) 
${\cal C}$ on $V$ if $U$ intersects every set in ${\cal C}$;
if ${\cal C}$ is a biset-family then $U$ should intersect the inner part of every member of ${\cal C}$.
Given costs $\{c_v:v \in V\}$, let $t^*({\cal C})$ denote the minimum value of a fractional 
${\cal C}$-transversal, namely:
$$t^*({\cal C})=\min\{\sum_{v \in V} c_v x_v: 
x(C) \geq 1 \ \ \forall \hat{C} \in {\cal C}, \ x(v) \geq 0 \ \forall v \in V\} \ .$$
\end{definition}

In \cite{N-aug}, the following is proved.

\begin{theorem} [\cite{N-aug}] \label{t:2}
let ${\cal C}$ be the family  of the inclusion members of a biset family ${\cal F}$.
Then the maximum degree in the hypergraph $\{C:\hat{C} \in {\cal C}\}$ is at most:
\begin{itemize}
\item[{\em (i)}]
${(4\gamma_{\cal C}+1)}^2$ if ${\cal F}$ is $D$-uncrossable.
\item[{\em (ii)}]
$2\gamma_{\cal C}+1$ if ${\cal F}$ is $T$-uncrossable.
\end{itemize}
\end{theorem}

Given a hypergraph $(V,{\cal C})$ with node-costs, the greedy algorithm computes in polynomial time 
a ${\cal C}$-transversal $U \subseteq V$ of cost  
$c(U) \leq H(\Delta({\cal C})) t^*({\cal C})$, where $\Delta({\cal C})$
is the maximum degree of the hypergraph (c.f. \cite{Lov}). 

\begin{lemma} \label{l:ratio}
If an edge-set $I$ covers a biset-family ${\cal F}$ then the set of endnodes of $I$
is a transversal of ${\cal F}$.
\end{lemma}

\begin{lemma} \label{l:feasible}
Let ${\cal F}$ be a biset family on $V$ and $I$ a star with center $a$ on a transversal 
$U \subseteq V$ of ${\cal F}$. Then $I$ covers ${\cal F}$ in each one of the following cases.
\begin{itemize}
\item[{\em (i)}]
${\cal F}$ is symmetric and $a \notin \Gamma(\hat{X})$ for all $\hat{X} \in {\cal F}$.
\item[{\em (ii)}]
$a \in V \setminus X^+$ for all $\hat{X} \in {\cal F}$.
\end{itemize}
\end{lemma}
\begin{proof}
Let $\hat{X} \in {\cal F}$. Then $a \in X$ or $a \in V \setminus X^+$.
If $a \in V \setminus X^+$, then since $U$ is a transversal of ${\cal C}$, there is $u \in U \cap X$. 
If $a \in X$, then if ${\cal F}$ is symmetric, then there $u \in U \cap (V \setminus X^+)$.
In both cases, there is an edge $au \in I$, and this edge covers $\hat{X}$. 
\qed
\end{proof}

The algorithm as in Theorem~\ref{t:main}, 
for both edge-costs and node-costs is as follows, 
where in the case of node-costs we may assume that the cost of $a$ is zero. 

\begin{center} 
\fbox{
\begin{minipage}{0.960\textwidth}
\begin{enumerate}
\item
For every $v \in V \setminus \{a\}$, 
let $e_v$ be the minimum-cost edge incident to $v$, and in the case of edge-costs 
define node-costs $c_v=\min_{e \in \delta_F(v)} c_e$ if $\delta_F(v) \neq \emptyset$,
and $c_v=\infty$ otherwise.
\item
Let ${\cal C}$ be the family of the inclusion members of ${\cal F}$.
With node-costs $\{c_v:v \in V\}$, compute a transversal $U$  
of ${\cal C}$ of cost $c(U) \leq H(\Delta({\cal C})) t^*({\cal C})$.
\item
Return $I=\{e_v: v\in U\}$.
\end{enumerate}
\end{minipage}
}
\end{center}

The solution computed is feasible by Lemma~\ref{l:feasible}.
The approximation ratio follows from Theorem~\ref{t:2} and Lemma~\ref{l:ratio}. \\

% The proof of Theorem~\ref{t:D} is complete.  

\section{Proof of Theorem~\ref{t:hard}} \label{s:hard}

Note that in the reduction in Observation~\ref{o:reduction} we have the following.
\begin{itemize}
\item
Uniform demands $d_v=k$ for all $v \in V$ in {\sf Survivable SL} correspond to
requirements $r_{sv}=k$ for all $v \in V \setminus \{s\}$ in {\sf Survivable {\cna}}.
\item
{\sf $\kappa'$-SL} with $p \equiv 1$ corresponds to {\sf Survivable {\cna}} with edge costs.
\item
Unit node-costs in {\sf Survivable SL} correspond to unit node-costs in {\sf Survivable {\cna}}.
\end{itemize}

Directed {\sf Rooted Survivable NA} with edge-costs and requirements 
$r_{sv}=k$ for all $v \in V \setminus \{s\}$ can be solved in polynomial time \cite{FT};
this implies that also {\em undirected} {\sf Survivable {\cna}} with edge-costs 
and requirements $r_{sv}=k$ for all $v \in V \setminus \{s\}$ can be solved in polynomial time.
Thus the same holds for {\sf $\kappa'$-SL} with $p \equiv 1$ and uniform demands.

Frank \cite{Frank} showed that {\em directed} 
{\sf Survivable {\cna}} with $\delta_G(s)=\emptyset$ and $k=1$ is NP-hard.
Using a slight modification of his reduction we can show that the problem is in fact 
{\sf Set-Cover} hard to approximate, and thus is $\Omega(\log n)$-hard to approximate. 
Given an instance of {\sf Set-Cover}, where a family $A$ of sets needs to cover a set $B$ of elements,
construct the corresponding directed bipartite graph  $G'=(A \cup B,E')$,
by putting an edge from every set to each element it contains.
The graph $G=(V,E)$ is obtained from $G'$ by adding $M$ copies of $B$, 
connecting $A$ to each copy in the same way as to $B$, and adding a new node $s$. 
Let $F=\{sv:v \in V\}$, $c(e)=1$ for every $e \in F$, and $r_{sv}=0$ if $v \in A$ and $r_{sv}=1$ otherwise. 
It is easy to see that if $I \subseteq F$ is a feasible solution to the obtained 
{\sf Survivable {\cna}} instance, then either $I$ corresponds to a feasible solution to
the {\sf Set-Cover} instance, or $|I| \geq M$.
The $\Omega(\log n)$-hardness follows for $M$ large enough, say $|M|={(|A|+|B|)}^2$, and $|A|=|B|$.
Since for $k=1$ all connectivity functions of {\sf Survivable NA} are equivalent,
we get $\Omega(\log n)$ hardness for directed {\sf Survivable NA} with $k=1$ and unit costs. 

\section{Proof of Theorem~\ref{t:fc-bounds}} \label{s:fc-bounds}

{\sf Survivable SL with Flow-Cost Bounds} is a special case of the following generalization
of the {\sf Submodular Cover} problem, where we have two progress functions:
\begin{equation} \label{e:fg}
f(S)= \sum_{v \in V}\min\{\lambda_G^{p,q}(S,v),d_v\}  \ \mbox{ and } \ 
  g(S)= \sum_{v \in V}\min\{-\mu_G^{p,q}(S,v),-b_v\} .
\end{equation}
It is easy to see that $S$ is a feasible solution to {\sf Submodular SL with Flow-Cost Bounds} if and only if 
both 
$$f(S)=f(V)=\sum_{v \in V} d_v \ \mbox{ and } \ 
  g(S)=g(V)=-\sum_{v \in V} b_v \ .$$
For $f,g$ defined by (\ref{e:fg}) we have 
$\max_{u \in U} f(\{u\})-f(\emptyset) \leq d(V)$,
but note that $\max_{u \in U} g(\{u\})-g(\emptyset)=\infty$ may hold.
The function $f$ is submodular since for any $v \in V$ the function $f_v(S)=\lambda^{p,q}_G(S,v)$
is submodular, as can be deduced from Lemma~\ref{l:star} and Theorem~\ref{t:sur}.
The function $g$ is submodular since for any $v \in V$ the function $g_v(S)=\lambda^{p,q}_G(S,v)$
is submodular; this is proved in \cite{BKP} for the case of edge-connectivity, and the proof
for $(p,q)$-connectivity is similar. Also, both functions are non-decreasing and 
admit a polynomial time value oracle.

\begin{center} 
\fbox{
\begin{minipage}{0.960\textwidth}
\noindent
{\sf Double Submodular Cover} \\
{\em Instance:}  
A groundset $V$ with costs $\{c_v:v \in V\}$ and submodular non-decreasing functions 
$f:2^V \rightarrow \mathbb{Z}$ and $g:2^V \rightarrow \mathbb{Z} \cup \{-\infty\}$ 
given by a value oracle. \\
{\em Objective:} 
Find $S \subseteq V$ of minimum cost with $f(S)=f(V)$ and $g(S)=g(V)$.
\end{minipage}
}
\end{center}

There are several natural approaches to solve the {\sf Double Submodular Cover} problem 
using the greedy algorithm of Wolsey \cite{W}.
One is to apply the greedy algorithm with the function $f+g$.
Another possibility is to solve two instances of {\sf Submodular Cover},
one with function $f$ and the other with function $g$, returning the union of the 
solutions $S_f$ and $S_g$ computed. 
However, in both cases the ratio may be unbounded if $g(\emptyset)=-\infty$, 
which may happen for $g$ defined by (\ref{e:fg}).

The idea is to compute $S_f$ and then to compute $S_g$ for the residual problem.
Note that for $f,g$ defined by (\ref{e:fg}) we have the following property:
if $f(S_f)=f(U)$ then $g(S) \geq - n \cdot c(E)$ for any $S \supseteq S_f$. 
Therefore, the following approach works. 
We take the set $S_f$ into our solution, and consider the residual {\sf Submodular Cover} problem 
with groundset $V \setminus S_f$ and the set function $h(S)=g(S_f \cup S)$, $S \subseteq V \setminus S_f$.
The function $h$ is submodular if $g$ is.
Note that for $g$ defined by (\ref{e:fg}),
$\max_{u \in U} h(\{u\})-h(\emptyset) \leq n \cdot c(E)-b(V)$, 
and we get approximation ratio 
$$H\left(\max_{v \in V} f(\{v\})-f(\emptyset)\right)+H\left(\max_{v \in V} h(\{v\})-h(\emptyset)\right)
\leq H(d(V))+H(nc(E)-b(V)) \ .$$

Clearly, the approach described can be generalized to the case when we have 
many non-decreasing submodular functions, under the assumption 
that there exists an ordering $f_1,f_2, \ldots$ of the functions  
such that for any $i$, if $f_j(S)=f(U)$ for every $j \leq i$, 
then $f_{j+1}(S') \neq -\infty$ for any $S' \supseteq S$. 

\vspace{0.3cm}

\noindent
{\bf Acknowledgment} The second author thank Takuro Fukunaga and an anonymous referee for many useful comments.

% \bibliographystyle{abbrv}
% \bibliography{SN-s}

\begin{thebibliography}{10}

\bibitem{AIMF}
K.~Arata, S.~Iwata, K.~Makino, and S.~Fujishige.
\newblock Locating sources to meet flow demands in undirected networks.
\newblock {\em J.~Algorithms}, 42:54–--68, 2002.

\bibitem{BKP}
J.~Bar-Ilan, G.~Kortsarz, and D.~Peleg.
\newblock Generalized submodular cover problems and applications.
\newblock {\em Theor. Comput. Sci. (TCS)}, 250(1-2):179--200, 2001.

\bibitem{CK-new}
J.~Chuzhoy and S.~Khanna.
\newblock An ${O}(k^3 \log n)$-approximation algorithms for vertex-connectivity
  survivable network design.
\newblock In {\em FOCS}, pages 437--441, 2009.

\bibitem{Frank}
A.~Frank.
\newblock Augmenting graphs to meet edge-connectivity requirements.
\newblock {\em SIAM Journal on Discrete Mathematics}, 5(1):25--53, 1992.

\bibitem{FT}
A.~Frank and E.~Tardos.
\newblock An application of submodular flows.
\newblock {\em Linear Algebra and its Applications}, 114/115:329--348, 1989.

\bibitem{Fuk}
T.~Fukunaga.
\newblock Approximating minimum cost source location problems with local
  vertex-connectivity demands.
\newblock In {\em TAMC}, pages 428--439, 2011.

\bibitem{Ishii}
T.~Ishii.
\newblock Greedy approximation for source location problem with
  vertex-connectivity requirements in undirected graphs.
\newblock {\em Journal of Discrete Algorithms}, 7:570--578, 2009.

\bibitem{IMAH}
H.~Ito, K.~Makino, K.~Arata, S.~Honami, Y.~Itatsu, and S.~Fujishige.
\newblock Source location problem with flow requirements in directed networks.
\newblock {\em Optimization Methods and Software}, 18(4):427–--435, 2003.

\bibitem{Jain}
K.~Jain.
\newblock A factor 2 approximation algorithm for the generalized steiner
  network problem.
\newblock {\em Combinatorica}, 21(1):39--60, 2001.

\bibitem{KN-sur}
G.~Kortsarz and Z.~Nutov.
\newblock {\em Approximating minimum cost connectivity problems, {\em Ch. 58
  in} Approximation Algorithms and Metahueristics, {\em T.~F.~Gonzales ed.}}
\newblock PWS, 2007.

\bibitem{KN-aug}
G.~Kortsarz and Z.~Nutov.
\newblock Tight approximation algorithm for connectivity augmentation problems.
\newblock {\em J. Computer and System Sciences}, 74(5):662--670, 2008.

\bibitem{Laek}
B.~Laekhanukit.
\newblock Parameters of two-prover-one-round game and the hardness of
  connectivity problems.
\newblock In {\em SODA}, pages 1626--1643, 2014.

\bibitem{Lov}
L.~Lov\'{a}sz.
\newblock On the ratio of optimal integral and fractional covers.
\newblock {\em Discrete Mathematics}, 13:383--390, 1975.

\bibitem{NI}
H.~Nagamochi and T.~Ibaraki.
\newblock {\em Algorithmic Aspects of Graph Connectivity, Chapter 9}.
\newblock Cambridge University Press, 2008.

\bibitem{NII}
H.~Nagamochi, T.~Ishii, and H.~Ito.
\newblock Minimum cost source location problem with vertex-connectivity
  requirements in digraphs.
\newblock {\em Information Processing Letters}, 80:287–--294, 2001.

\bibitem{N-Focs}
Z.~Nutov.
\newblock Approximating minimum cost connectivity problems via uncrossable
  bifamilies.
\newblock {\em Transactions on Algorithms}, 9(1):1, 2012.

\bibitem{N-aug}
Z.~Nutov.
\newblock Approximating node-connectivity augmentation problems.
\newblock {\em Algorithmica}, 63(1-2):398--410, 2012.

\bibitem{SMF}
M.~Sakashita, K.~Makino, and S.~Fujishige.
\newblock Minimum cost source location problems with flow requirements.
\newblock {\em Algorithmica}, 50:555--583, 2008.

\bibitem{Sch}
A.~Schrijver.
\newblock {\em Combinatorial Optimization Polyhedra and Efficiency}.
\newblock Springer, 2004.

\bibitem{TSSA}
H.~Tamura, M.~Sengoku, S.~Shinoda, and T.~Abe.
\newblock Location problems on undirected flow networks.
\newblock {\em IEICE Trans.}, E73:1989–--1993, 1990.

\bibitem{Vakilian}
A.~Vakilian.
\newblock Node-weighted prize-collecting network design problems.
\newblock MS Thesis, 2013.

\bibitem{W}
L.~A. Wolsey.
\newblock An analysis of the greedy algorithm for the submodular set covering
  problem.
\newblock {\em Combinatorica}, 2:385–--393, 1982.

\end{thebibliography}

\end{document}